\documentclass[12pt,reqno]{article}

\usepackage[usenames]{color}
\usepackage{xcolor}
\usepackage{amsthm}
\usepackage{amsmath,amssymb}
\usepackage{graphicx}

\usepackage{amscd}
\usepackage[colorlinks=true,
linkcolor=webgreen,
filecolor=webbrown,
citecolor=webgreen]{hyperref}

\definecolor{webgreen}{rgb}{0,.5,0}
\definecolor{webbrown}{rgb}{.6,0,0}

\usepackage{color}
\usepackage{fullpage}
\usepackage{enumerate}
\usepackage{float}

\usepackage{amsfonts}
\usepackage{latexsym}
\usepackage{graphics}
\usepackage{url}

\usepackage{bm}

\usepackage{amsthm, amsmath, amsfonts, amssymb, mathtools, 
  float, hyperref, enumitem, pifont, mathrsfs, tabularx, mathabx,
  tikz, microtype}
\usepackage[normalem]{ulem}
\usetikzlibrary{decorations.markings, decorations.pathmorphing}

\theoremstyle{plain}
\newtheorem{thm}{Theorem}
\newtheorem{cor}[thm]{Corollary}

\newtheorem{prop}[thm]{Proposition}

\theoremstyle{definition}
\newtheorem{defn}[thm]{Definition}

\newtheorem{por}[thm]{Porism}

\theoremstyle{remark}
\newtheorem{rem}[thm]{Remark}

\DeclareMathAlphabet{\mymathbb}{U}{BOONDOX-ds}{m}{n}
\newcommand{\pp}{\mathcal{P}}

\newcommand{\mm}{\mathcal{M}}
\newcommand{\ind}{\mymathbb{1}}
\newcommand{\zz}{\mathbb{Z}}
\newcommand{\st}{\mathbf}
\renewcommand{\aa}{\mathcal{A}}
\DeclareRobustCommand{\inc}{{\tikz{\draw[line width = 0.61pt, scale = 0.13, ->] (0,0) -- (1,1)}}}

\newcommand\oeis[1]{\href{http://oeis.org/#1}{\underline{#1}}}

\let\ge\geqslant
\let\leq\leqslant
\let\geq\geqslant

\begin{document}

\date{\today}

\title{\LARGE\bf Enumeration of  Dyck paths with air pockets}
\author{\large Jean-Luc {Baril}, Sergey Kirgizov, R{\'e}mi Mar{\'e}chal and Vincent Vajnovszki\\
LIB, Université de Bourgogne\\
B.P. 47 870, 21078 DIJON-Cedex France      \\
\href{mailto:barjl@u-bourgogne.fr}{\tt barjl@u-bourgogne.fr},
\href{mailto:sergey.kirgizov@u-bourgogne.fr}{\tt sergey.kirgizov@u-bourgogne.fr}\\
\href{mailto:remi\_marechal01@etu.u-bourgogne.fr}{\tt remi\_marechal01@etu.u-bourgogne.fr}, \href{mailto:vvajnov@u-bourgogne.fr}{\tt vvajnov@u-bourgogne.fr}
}
\maketitle

\vskip .2 in

\maketitle

\begin{abstract} We introduce
and study the new  combinatorial class 
of  Dyck paths with air pockets.
 We exhibit a bijection  with  the 
 peakless Motzkin paths which transports 
 several pattern statistics and give 
 bivariate generating functions for the 
 distribution of patterns as peaks, 
 returns and pyramids. Then, we deduce 
 the popularities  and asymptotic expectations
 of these patterns and 
 point out a link between the popularity 
 of pyramids and a special kind of
closed smooth self-overlapping curves, a subset of
Fibonacci meanders. 
A similar study is conducted for non-decreasing  Dyck paths with air pockets.

\end{abstract}


\section{Introduction and notations}
In combinatorics, lattice paths are widely studied. They have many applications in various domains such as computer science, biology and physics~\cite{Sta}, and  they have very tight links with other combinatorial objects such as directed animals, pattern avoiding  permutations,  bargraphs,  RNA structures  and so on~\cite{Bar1,Knu,Sta}. A classical problem in combinatorics is the enumeration of these paths with respect to their length and  other  statistics~\cite{Ban,Barc,Bar,Deu,Man1,Mer,Pan,Sap,Sun}. In the literature, Dyck and Motzkin paths are the most often considered. They  are  counted  by the  famous  Catalan and  Motzkin  numbers  (see \href{https://oeis.org/A000108}{\underline{A000108}} and \href{https://oeis.org/A001006}{\underline{A001006}} in Sloane's
On-line
Encyclopedia  of  Integer
Sequences~\cite{oeis}). In 2005, Dyck paths with catastrophes have been introduced by Krinik 
{\em et al.} in~\cite{Kri} in the context of queuing theory. They correspond to the evolution of a queue by allowing some resets. The push (resp. pop) operation corresponds to a step $U=(1,1)$ (resp. $D=(1,-1)$), and the reset operations are modeled by catastrophe steps $D_k=(1,-k)$ ending on $x$-axis, $k\geq 2$.
Banderier and Wallner~\cite{Ban} study these paths by providing enumerative results and limit laws.

In this paper, we introduce and study the paths with air pockets corresponding to a queue evolution with partial reset operations that cannot be consecutive. These paths can also be viewed as airplane flights, sometimes showcasing turbulences that are
known as {\it air pockets},  where consecutive turbulences are considered to be one. More formally, a \emph{Dyck path with air pockets} is a nonempty lattice path in the first quadrant of $\zz^2$ starting at the origin, ending on the $x$-axis, and consisting of up-steps $U=(1,1)$ and down-steps $D_k=(1,-k)$, $k\geq 1$, where two down steps cannot be consecutive. For short, we set $D=D_1$. The \emph{length} of a Dyck path with air pockets is the number of its steps. Let $\aa_n$ be the set of $n$-length Dyck paths with air pockets. By definition $\aa_0=\aa_1=\varnothing$ and 
 we set $\aa=\bigcup_{n\geq 2}\aa_n$.
 
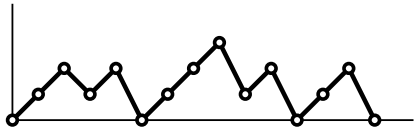
\begin{figure}[h]
\begin{center}
\scalebox{0.86}{\begin{tikzpicture}[ultra thick]
   \draw[black, thick] (0,0)--(6.2,0);\draw[black, thick] (0,0)--(0,1.8);
   \draw[black, line width=2pt](0,0)--(0.4,0.4)--(0.8,0.8)--(1.2,0.4)--(1.6,0.8)--(2,0)--(2.4,0.4)--(2.8,0.8)--(3.2,1.2)--(3.6,0.4)--(4,0.8)--(4.4,0)--(4.8,0.4)--(5.2,0.8)--(5.6,0);

   \tikzset{every node/.style={circle, draw,fill=white,inner sep=1.5pt}}
   \node at (0,0) {};
   \node at (0.4,0.4) {};
   \node at (0.8,0.8) {};
   \node at (1.2,0.4) {};
   \node at (1.6,0.8) {};
   \node at (2,0) {};
   \node at (2.4,0.4) {};
   \node at (2.8,0.8) {};
   \node at (3.2,1.2) {};
   \node at (3.6,0.4) {};
   \node at (4,0.8) {};
   \node at (4.4,0) {};
   \node at (4.8,0.4) {};
    \node at (5.2,0.8) {};
    \node at (5.6,0) {};
\end{tikzpicture}}
\end{center}
\caption{The Dyck path with air pockets $UUDUD_2UUUD_2UD_2UUD_2$.}
\label{fig1}
\end{figure}

A Dyck path with air pockets is called \emph{prime} whenever it ends with $D_k$, $k\geq 2$, and
returns to the $x$-axis only once.
The set of all prime  Dyck paths with air pockets of length $n$ is denoted $\pp_n$.
Notice that $UD$ is not prime so we set $\pp=\bigcup_{n\geq 3}\pp_n$.
If $\alpha=U\beta UD_k\in\pp_n$, then $2\leq k<n$, $\beta$ is a (possibly empty) prefix of a path in $\mathcal{A}$, and we define the Dyck path with air pockets $\alpha^\flat=\beta UD_{k-1}$, called the `lowering' of $\alpha$. For example, the path $\alpha=UUDUUD_3$ is prime, and $\alpha^\flat=UDUUD_2$. The map $\alpha\mapsto\alpha^\flat$ is clearly a bijection from $\pp_n$ to $\aa_{n-1}$ for all $n\geq 3$, and we denote by $\gamma^\sharp$ the inverse image of $\gamma\in \mathcal{A}_{n-1}$ ($\alpha^\sharp$ is a kind of `elevation' of $\alpha$, drawing inspiration for the term from Deutsch's definition of elevated Dyck paths~\cite{Deu}).
Any Dyck path with air pockets $\alpha\in\aa$ can be decomposed depending on its \emph{second-to-last return to the $x$-axis}: either ($i$) $\alpha=UD$, or  ($ii$) $\alpha=\beta UD$ with $\beta\in\aa$, or ($iii$) $\alpha\in \pp$, or $(iv)$  $\alpha=\beta\gamma$ where  $\beta\in\aa$ and $\gamma\in\pp$. So, if $A(x)=\sum_{n\geq 2}a_nx^n$ where $a_n$ is the cardinality of $\aa_n$, and  $P(x)=\sum_{n\geq 3}p_nx^n$ where $p_n$ is the cardinality of $\pp_n$, then we have $P(x)=xA(x)$ and the previous
decompositions imply the functional equation
$A(x)=x^2+x^2A(x)+xA(x)+xA(x)^2,$ and
\begin{equation}A(x)={\frac {1-x-{x}^{2}-\sqrt {{x}^{4}-2\,{x}^{3}-{x}^{2}-2\,x+1}}{2x
}}
\label{rel1}\end{equation}
which generates the generalized Catalan
numbers (see 
\href{https://oeis.org/A004148}{\underline{A004148}} 
in~\cite{oeis}), which among other things, 
counts the peakless Motzkin paths. The 
first values of $a_n$ for $2\leq n\leq 10$ 
are
$1,1,2,4,8,17,37,82,185$. An asymptotic approximation for the coefficient of $x^n$ in the series expansion of $A(x)$ is $$\frac{\sqrt{14\sqrt{5}-30}}{2n\sqrt{\pi n}(3-\sqrt{5})}\left(\frac{\sqrt{5}+3}{2}\right)^n.$$\par
If a Dyck path with air pockets 
$\alpha\in\aa_n$ has $k\geq 1$ peaks (a 
peak is an occurrence $UD_i$ for some 
$i\geq 1$), then it contains $n-k$ 
up-steps. If we `unfurl' all of its 
down-steps $D_i$, $i\geq 1$, into runs 
$D^i$ of $i$ consecutive $D$-steps, then we
obtain a Dyck path of length $2(n-k)$
having $k$ peaks. This gives rise to a
bijection between Dyck paths of semilength
$n-k$ with $k$ peaks and $n$-length Dyck
paths with air pockets with $k$ peaks.
Hence, the number of $n$-length Dyck paths
with air pockets with $k$ peaks is equal to
the Narayana number
$N(n-k,k)=\frac{1}{n-k}\binom{n-k}{k}\binom{n-k}{k-1}$
(see~\cite{Deu}).

In the following, a {\it pattern} consists of consecutive steps in a path, and a {\it statistic} $\bf{s}$ is an integer-valued function from a set
$\mathcal{S}$ of  paths. To a given pattern $p$, we associate the
pattern statistic ${\bf p} : \mathcal{S} \to \mathbb{N}$ where ${\bf p}(a)$ is the
number of occurrences of the pattern $p$ in  $a \in \mathcal{S}$ (we
use the boldface to denote statistics). For example,
  the statistic giving the number of
  occurrences of the consecutive pattern $UU$ in  a path is denoted by $\bf{UU}$.
  For $n\geq 1$, we denote by $\hat{\bf{n}}$ the constant statistic 
  returning the value $n$.
  The \emph{popularity} of a pattern $p$ in $\mathcal{S}$ is the total number of
  occurrences of $p$ over all objects of $\mathcal{S}$, that is
  ${\bf p}(\mathcal{S})=\sum_{a\in \mathcal{S}}{\bf p}(a)$ (\cite{ Bon, Homb, Kit}).
 Let $\mathcal{S}'$ be another set of combinatorial objects, we say that two statistics, $\bf{s}$ on $\mathcal{S}$ and $\bf{t}$ on $\mathcal{S}'$, have the \emph{same distribution} if there exists a bijection $f : \mathcal{S} \to \mathcal{S}'$ satisfying ${\bf s}(a) = {\bf t} (f (a))$ for any $a \in \mathcal{S}$. In this case, with a slight abuse of the notation already used in~\cite{patdd}, we write   $f(\bf{s})=\bf{t}$ or $\bf{s} =\bf{t}$ whenever $f$ is the identity.

The remainder of this paper is organized as follows. 
In Section 2, we present a constructive bijection between $n$-length Dyck paths with air pockets and peakless Motzkin paths of length $n-1$, and we show how this bijection transports some statistics. In Section 3,  we provide bivariate generating functions $A(x,y)=\sum_{n,k\geq 0}a_{n,k}x^ny^k$ for the distributions of some statistics $\st{s}$, i.e.\ the coefficient $a_{n,k}$ of $x^ny^k$ is the number of paths $\alpha\in \mathcal{A}_n$ satisfying $\st{s}(\alpha)=k$. Then, we deduce the popularities of some patterns ($U$, $D$, peak, return, catastrophe, pyramid, \ldots ) by calculating $\partial_y(A(x,y))\rvert_{y=1}$, and  we provide asymptotic approximations for them using classical methods (see~\cite{flajolet,orlov}). We refer to  Table 1 for an overview of the results. As a byproduct,  we point out a link between the popularity of pyramids and a special kind of closed smooth self overlapping curves in the plane (a subset of Fibonacci meanders defined
in~\cite{Luschny, Wienand}).
In Section 4, we make a similar study for non-decreasing Dyck paths with air pockets.

\begin{table}[H]
    \centering
    \begin{tabular}{lll}
    \hline
     Pattern & Pattern popularity in $\aa_n$ & OEIS\\
    \hline
    U&$1, 2, 5, 13, 32, 80, 201, 505, 1273, 3217 $&\href{http://oeis.org/A110320}{\underline{A110320}}\\
    D&$1, 0, 2, 3, 7, 17, 40, 97, 238, 587$&\href{http://oeis.org/A051291}{\underline{A051291}}\\
    Peak & $1, 1, 3, 7, 16, 39, 95, 233, 577, 1436$ & \href{http://oeis.org/A203611}{\underline{A203611}}\\
    Ret & $1, 1, 3, 6, 13, 29, 65, 148, 341, 793$ & \href{http://oeis.org/A093128}{\underline{A093128}}\\
    Cat & $0, 1, 1, 4, 8, 19, 44, 102, 239, 563$ & \\
    $\Delta_k$ & $\underbrace{0, \hdots, 0}_{k-1\text{ zeroes}}, 1, 0, 2, 3, 7, 17, 40, 97, 238, 587$ &  \href{http://oeis.org/A051291}{\underline{A051291}}\\
    $\Delta_{\geq k}$ & $\underbrace{0, \hdots, 0}_{k-1\text{ zeroes}}, 1, 1, 3, 6, 13, 30, 70, 167, 405$ & \href{http://oeis.org/A201631}{\underline{A201631}}($=u_n$)\\
    $\Delta_{\leq k}$ & $\Delta_{\leq 1}\quad1, 0, 2, 3, 7, 17, 40, 97, 238, 587$&$u_n-u_{n-k}$\\
    &$\Delta_{\leq 2}\quad1, 1, 2, 5, 10, 24, 47, 137, 335, 825, \hdots$&\\
    &$\Delta_{\leq 3}\quad1, 1, 3, 5, 12, 27, 64, 154, 375, 922, \hdots$&\\
    &$\vdots$&\\
    \hline
    \end{tabular}
    \caption{Pattern popularity in $\aa_n$, for $2\leq n\leq 11$.}
\end{table}

\section{Bijection with peakless Motzkin paths}

In this section we exhibit a constructive bijection between  $n$-length Dyck paths with air pockets and  $(n-1)$-length peakless Motzkin paths, i.e.\ lattice paths in the first quarter plane starting at the origin, ending at $(n-1,0)$, made of $U$, $D$ and $F=(1,0)$ and having no occurrence of $UD$. 
 Moreover, we show how our bijection transports some pattern based statistics. We denote by $\mm_n$ the set of peakless Motzkin paths of length $n$, and $\mm=\bigcup_{n\geq 1}\mm_n$.

\begin{defn} We  recursively define the map $\psi$ from $\aa$ to $\mm$ as follows. For $\alpha\in\aa$, we set:
$$
\psi(\alpha)=\left\{\begin{array}{llr}
F&\text{if }\alpha=UD,&(i)\\
U\psi(\beta)D&\text{if }\alpha=\beta UD\text{ with } \beta\in\aa,&(ii)\\
\psi(\alpha^\flat)F&\text{if }\alpha\in\pp,&(iii)\\
\psi(\gamma^\flat)U\psi(\beta)D&\text{if }\alpha=\beta\gamma \text{ with } \beta\in\aa \text{ and } \gamma\in\pp. &(iv)
\end{array}\right.
$$
\label{Def1}
\end{defn}
Notice that each factor in the above decomposition is nonempty, and that $\psi$ maps nonempty objects to nonempty ones. Due to the recursive definition, the image by $\psi$ of a $n$-length Dyck path with air pockets is a peakless Motzkin path of length $n-1$. For instance, the images of $UD$, $UUD_2$, $UUUD_2UD_2UD$ are respectively $F$, $FF$, and $UUFFDFD$. We refer to Figure~\ref{fig2} for an illustration of this mapping.
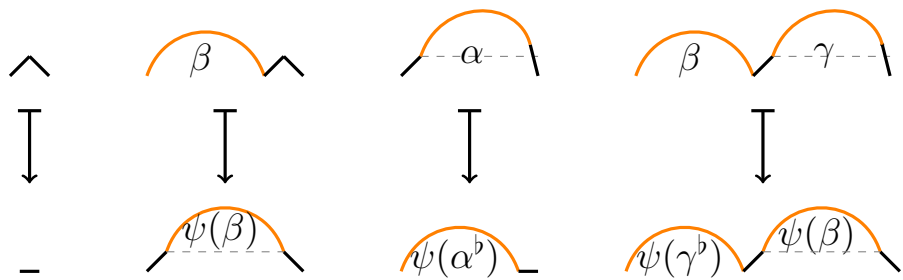
\begin{figure}[H]
\centering
\begin{tikzpicture}[scale=0.26]
\draw[very thick] (0,0) -- (1,1);
\draw[very thick] (1,1) -- (2,0);

\draw[very thick] (0.5,-10) -- (1.5,-10);

\draw[orange, very thick] (7,0) .. controls (8,3) and (12,3) .. (13,0) node[black,xshift=-24,yshift=6] {\large $\beta$};
\draw[very thick] (13,0) -- (14,1);
\draw[very thick] (14,1) -- (15,0);

\draw[orange, very thick] (8,-9) .. controls (9,-6) and (13,-6) .. (14,-9) node[black,xshift=-24,yshift=8] {\large $\psi(\beta)$};
\draw[very thick] (7,-10) -- (8,-9);
\draw[very thick] (14,-9) -- (15,-10);
\draw[gray, thin, dashed] (8,-9) -- (14,-9);

\draw[very thick] (20,0) -- (21,1);
\draw[gray, thin, dashed] (21,1) -- (27,1);
\draw[very thick] (26.6,1.6) -- (27,0);
\draw[orange, very thick] (21,1) .. controls (22,4) and (26,4) .. (26.6,1.6) node[black,xshift=-22,yshift=-4] {\large ${\color{black}\alpha}$};

\draw[orange, very thick] (20,-10) .. controls (21,-7) and (25,-7) .. (26,-10) node[black,xshift=-24,yshift=6] {\large $\psi(\alpha^\flat)$};
\draw[very thick] (26,-10) -- (27,-10);

\draw[orange, very thick] (32,0) .. controls (33,3) and (37,3) .. (38,0) node[black,xshift=-24,yshift=6] {\large $\beta$};
\draw[gray, thin, dashed] (39,1) -- (45,1);
\draw[very thick] (38,0) -- (39,1);
\draw[very thick] (44.6,1.6) -- (45,0);
\fill[black] (44.6,1.6) circle (1mm);
\fill[white] (41.8,0.5) rectangle (42.2,1.5);
\draw[orange, very thick] (39,1) .. controls (40,4) and (44,4) .. (44.6,1.6) node[black,xshift=-22,yshift=-4] {\large ${\color{black}\gamma}$};

\draw[orange, very thick] (31.5,-10) .. controls (32.5,-7) and (36.5,-7) .. (37.5,-10) node[black,xshift=-24,yshift=6] {\large $\psi(\gamma^\flat)$};
\draw[orange, very thick] (38.5,-9) .. controls (39.5,-6) and (43.5,-6) .. (44.5,-9) node[black,xshift=-24,yshift=6] {\large $\psi(\beta)$};
\draw[very thick] (37.5,-10) -- (38.5,-9);
\draw[very thick] (44.5,-9) -- (45.5,-10);
\draw[gray, thin, dashed] (38.5,-9) -- (44.5,-9);

\draw[very thick,|->] (1,-1.7) -- (1,-5.5);
\draw[very thick,|->] (11,-1.7) -- (11,-5.5);
\draw[very thick,|->] (23.5,-1.7) -- (23.5,-5.5);
\draw[very thick,|->] (38.5,-1.7) -- (38.5,-5.5);
\end{tikzpicture}
\caption{Illustration of the map $\psi$ according to Definition~\ref{Def1}.}
\label{fig2}
\end{figure}

\begin{thm}
For all $n\geq 2$, the map $\psi$ induces a bijection between $\aa_n$ and $\mm_{n-1}$.
\label{bij}
\end{thm}
\begin{proof}
 It is well known that the cardinality of $\mm_n$ is given by the $n$-th term of generalized Catalan number (see \href{https://oeis.org/A004148}{\underline{A004148}} in~\cite{oeis}). So it suffices (see observation after relation~(\ref{rel1})) to prove the injectivity of $\psi$. We proceed by induction on $n$. The case $n=2$ is obvious since $\aa_2=\{UD\}$ and $\mm_1=\{F\}$.
 For all $k\leq n$, we assume that $\psi$ is an injection from $\aa_k$ to $\mm_{k-1}$, and we prove the result for $k=n+1$.
 According to Definition~\ref{Def1}, if $\alpha$ and $\beta$ in $\aa_{n+1}$ satisfy $\psi(\alpha)=\psi(\beta)$, then $\alpha$ and $\beta$ necessarily come from the same case among ($i$)~--~($iv$). Using the induction hypothesis, we conclude directly that $\alpha=\beta$, which completes the induction. Thus $\psi$ is injective and so bijective.
\end{proof}

\begin{prop}
For all $n\geq 2$ and $k\geq 1$, and $\psi : \aa_n\rightarrow\mm_{n-1}$, the following holds:

\begin{itemize}[label=\ding{0}]
    \item[$\bullet$] $\psi(\st{U})=\st{F}+\st{U}=\st{F}+\st{D}$
    \item[$\bullet$] $\psi(\st{D})=\psi(\st{UD})=\st{\ind_{F}}+\st{UFD}+\st{\ind_{U\mm D}}+\st{U^2\mm D^2}$
    \item[$\bullet$] $\psi(\st{DU})=\st{UFD}+\st{U^2\mm D^2}$
    \item[$\bullet$] $\psi(\st{UU})=\st{F}-\st{\hat{1}}$
    \item[$\bullet$] $\psi(\st{\Delta_k})=\st{\ind_{F^k}}+\st{UF^kD}+\st{\ind_{F^{k-1}U\mm D}}+\st{UF^{k-1}U\mm D^2}$
    \item[$\bullet$] $\psi(\st{Peak})=\st{U}+\st{\hat{1}}$
    \item[$\bullet$] $\psi(\st{Ret})=\st{\hat{n}}-\st{LastF}$
    \item[$\bullet$] $\psi(\st{SLast})=\st{Ret}$,
\end{itemize}
where $\st{\ind_\beta}(\alpha)=1$ if $\alpha=\beta$ and $0$ otherwise; $\st{\ind_{U\mm D}}(\alpha)$ is equal to $1$ if there exists $\beta\in\mm$ such that $\alpha=U\beta D$ and $0$ otherwise; $\st{U^2\mm D^2}(\alpha)$ is the number of occurrences  $U^2\beta D^2$ in $\alpha$ for ${\beta\in\mm}$;
$\st{\Delta_k}(\alpha)$ is the number of occurrences $U^kD_k$ in $\alpha$;
$\st{Peak}(\alpha)=\st{\sum_{k\geq 1}\st{UD_k}(\alpha)}$; $\st{Ret}(\alpha)$ is the number of returns to the $x$-axis of $\alpha$; $\st{LastF}(\alpha)$ is the position of the rightmost flat-step in $\alpha$, and $\st{SLast}(\alpha)$ is the size of the  the last step of $\alpha$ (i.e.\ $k$ if the last step is $D_k$).
\end{prop}

\begin{proof} We provide the proof for $\psi(\st{U})$ and $\psi(\st{\Delta_k})$ since  those for the other relations can be obtained {\it mutatis mutandis}.

We proceed by induction on $n$. Since $\aa_2=\{UD\}$ and  $\mm_1=\{F\}$ the statements trivially hold for $n=2$. Now, assume the statements are true for all $k\leq n$ and let us prove them for $n+1$.

If $\alpha\in\aa_{n+1}$ with $n\geq 2$, then we have either ($i$) $\alpha=\beta UD$, ($ii$) $\alpha=\gamma^\sharp$ or ($iii$) $\alpha=\beta\gamma^\sharp$ where $\beta,\gamma\in \aa$.
In the case ($i$), $\st{U}(\beta UD)=1+\st{U}(\beta)$ and with the induction hypothesis,  $\st{U}(\beta UD)=1+\st{(U+F)}(\psi(\beta))=\st{(U+F)}(U\psi(\beta)D)=\st{(U+F)}(\psi(\beta UD))$ as expected. In the case ($ii$), $\st{U}(\gamma^\sharp)=1+\st{U}(\gamma)$ and with the induction hypothesis,  $\st{U}(\gamma^\sharp)=1+\st{(U+F)}(\psi(\gamma))=\st{(U+F)}(\psi(\gamma)F)=\st{(U+F)}(\psi(\gamma^\sharp))$. Case ($iii$) is handled in the same way.

So, we have $\psi(\st{U})=\st{U}+\st{F}$. Using a similar reasoning, we can easily prove $\psi(\st{D})=\st{\ind_{F}}+\st{UFD}+\st{\ind_{U\mm D}}+\st{U^2\mm D^2}$.
\par
Now, let us give details for the slightly less straightforward case of $\psi(\st{\Delta_k})$ for $k\geq 1$.  The case $k=1$ is already handled since we have $\psi(\st{UD})=\psi(\st{D})$. So, we assume $k\geq 2$.
We consider the following case analysis: any given Dyck path with air pockets is either of the form ($i$)~$\beta UD$, ($ii$)~$\beta\Delta_{k-1}^\sharp$, ($iii$)~$\beta\Delta_k^\sharp$, ($iv$)~$\beta(\alpha\Delta_k)^\sharp$ with $\alpha\in\aa$, or ($v$)~$\beta\alpha^\sharp$ with $\alpha\in\aa$ being neither $\Delta_{k-1}$, nor $\Delta_k$, nor $\alpha'\Delta_k$ ($\alpha'\in\aa$), and $\beta\in\aa\cup\{\varepsilon\}$. Reasoning by induction, case ($ii$) unfolds as follows: if $\beta=\varepsilon$, then we get 
$$
(\st{\ind_{F^k}}+\st{UF^kD}+\st{\ind_{F^{k-1}U\mm D}}+\st{UF^{k-1}U\mm D^2})(F^k)=1,
$$
which is the same as $\st{\Delta_k}(\Delta_{k-1}^\sharp)$. Otherwise, we have
\begin{align*}
&(\st{\ind_{F^k}}+\st{UF^kD}+\st{\ind_{F^{k-1}U\mm D}}+\st{UF^{k-1}U\mm D^2})(\psi(\beta\Delta_{k-1}^\sharp))=\\
&\;=(\st{\ind_{F^k}}+\st{UF^kD}+\st{\ind_{F^{k-1}U\mm D}}+\st{UF^{k-1}U\mm D^2})(F^{k-1}U\psi(\beta)D)=\\
&\;=\st{\ind_{F^k}}(\psi(\beta))+\st{UF^kD}(\psi(\beta))+1+\st{\ind_{F^{k-1}U\mm D}}(\psi(\beta))+\\
&\;\;\;\;\;+\st{UF^{k-1}U\mm D^2}(\psi(\beta))=\\
&\;=1+\st{\Delta_k}(\beta)=\st{\Delta_k}(\beta\Delta_{k-1}^\sharp).
\end{align*}

The four remaining cases are obtained in the same way.
\end{proof}

Notice that the mirror of a Dyck path with air pockets is a \L{}ukasiewicz path avoiding flat steps and two consecutive up-steps. Since there is a bijection between  \L{}ukasiewicz paths and plane trees (see \cite{Kor} for instance), we easily deduce that Dyck paths with air pockets are in one-to-one correspondence with plane trees without unary nodes, and such that the first child of any node is always a leaf. We leave open the question of knowing how this bijection transports some pattern-based statistics.

\section{Distribution  and popularity of patterns}

\subsection{The numbers of
\texorpdfstring{$U$}{U} and
\texorpdfstring{$D$}{D}}
\begin{thm}
Let $A(x,y,z)=\sum_{n,k,\ell\geq 0}a_{n,k,\ell}x^ny^kz^\ell$ be the generating function (g.f.) where $a_{n,k,\ell}$ is the number of paths in $\aa_n$ having $k$ up-steps $U$ and $\ell$ down-steps $D=D_1$. Then the following holds:
$$
A(x,y,z)=\frac {1-xy-{x}^{2}yz-2\,{x}^{3}{y}^{2}+2\,{x}^{3}{y}^{2}z-
\sqrt {R}}{2xy \left( 1+{x}^{2}y-{x}^{2}yz
 \right) },
 $$
with
$$
R={x}^{4}{y}^{2}{z}^{2}+2\,{x}^{3}{y}^{2}z-4\,{x}^{3}{y}^{2}+{x}^
{2}{y}^{2}-2\,{x}^{2}yz-2\,xy+1.
$$
\end{thm}

\begin{proof}
Due to the first return decomposition, any Dyck path with air pockets has one of the following forms: $(i)$~$UD\gamma$, $(ii)$ $U^2D_2\gamma$, $(iii)$ $(\alpha UD)^\sharp\gamma$ with $\alpha\in\aa$, $(iv)$ $\alpha^\sharp\gamma$ with $\alpha\in\aa$ not being  $UD$ nor $\beta UD$ ($\beta\in\aa$), where $\gamma\in \aa\cup\{\varepsilon\}$. These four cases are disjoint and cover $\aa$ entirely. Then, we deduce the  functional equation by taking into account the length, and the numbers of $U$ and $D$ with respect to $x$, $y$ and $z$:
$$
A=\left(\underbrace{\vphantom{xy\left(A-x^2yz-x^2y^2zA\right)}x^2yz}_{(i)}+\underbrace{\vphantom{xy\left(A-x^2yz-x^2y^2zA\right)}x^3y^2}_{(ii)}+\underbrace{\vphantom{xy\left(A-x^2yz-x^2y^2zA\right)}x^3y^2A}_{(iii)}+\underbrace{xy\left(A-x^2yz(1+A)\right)}_{(iv)}\right)(1+A),
$$
where $A$ stands for $A(x,y,z)$. Solving for $A$, we get the result.
\end{proof}

\begin{cor}
For all $n\geq 1$, the number of Dyck paths with air pockets (of any length) having $n$ up-steps $U$ is the $n$-th Catalan number $\frac{1}{n+1}\binom{2n}{n}$ (see  \href{http://oeis.org/A000108}{\underline{A000108}} in~\cite{oeis}).
\end{cor}

\begin{proof}
We check that $1+A(1,y,1)$ is the g.f.\ of the Catalan numbers.
\end{proof}

\begin{cor}
For all $n\ge 1$, the number of Dyck paths with air pockets having  $n$ up-steps $U$ and no down-steps $D$ is  the $n$-th Riordan number $\sum_{k=0}^n(-1)^{n-k} \binom{n}{k} c_k,$ where $c_k=\frac{1}{k+1}\binom{2k}{k}$ 
(see \href{http://oeis.org/A005043}{\underline{A005043}} in~\cite{oeis}).
\end{cor}

\begin{proof}
We check that $1+A(1,y,0)$ is the g.f.\ of the Riordan numbers.
\end{proof}

\begin{cor}
The g.f.\ for the popularity of up-steps $U$ in $\aa_n$ is
$$
\frac{1-x-x^2-\sqrt{x^4-2x^3-x^2-2x+1}}{2x\sqrt{x^4-2x^3-x^2-2x+1}},
$$
which generates a shift of the sequence \href{https://oeis.org/A110320}{\underline{A110320}} in~\cite{oeis}. An asymptotic approximation of the $n$-th term is $$\frac{\sqrt{5}-1 }{2\sqrt{\pi n}\, \sqrt{14 \sqrt{5}-30} } \left(\frac{3+\sqrt{5}}{2}\right)^n,$$
and an asymptotic for the expectation of the up-step number is $$\frac{\sqrt{5}+5}{10}n\sim 0.723606799\cdot n.$$
\end{cor}

\begin{proof}
The g.f.\ is given by $\partial_y(A(x,y,1))\rvert_{y=1}$. The asymptotic approximation is obtained using classical methods (see~\cite{flajolet,orlov})
\end{proof}

\begin{cor}
The g.f.\ for the popularity of down-steps $D=D_1$ in $\aa_n$ is
$$
\frac{x^2\left(1+2x^2-x^3+(1-x)\sqrt{x^4-2x^3-x^2-2x+1}\right)}{2\sqrt{x^4-2x^3-x^2-2x+1}},
$$
which generates a shift of the sequence \href{https://oeis.org/A051291}{\underline{A051291}} in~\cite{oeis}. An asymptotic approximation of the $n$-th term is $$\frac{5\sqrt{5}-11}{2 \sqrt{\pi n}\, \sqrt{14 \sqrt{5}-30}}\left(\frac{3+\sqrt{5}}{2}\right)^n
,$$
and an asymptotic for the expectation of the $D$-step number is 
$$\frac{5-2\sqrt{5}}{5}n\sim 0.105572797\cdot n.$$
\end{cor}

\begin{proof}
The g.f.\ is given by $\partial_z(A(x,1,z))\rvert_{z=1}$.
\end{proof}

\subsection{The number of peaks}
In this part, we study the distribution of peaks, i.e.\ patterns $UD_m$ for $m\geq 1$.
\begin{thm}
Let $P(x,y)=\sum_{n,k\geq 0} p_{n,k}x^ny^k$ be the g.f.\  where $p_{n,k}$ is the number of $n$-length Dyck paths with air pockets having $k$ peaks. Then we have:
$$
P(x,y)=\frac{1-x-x^2y-\sqrt{(1-x-x^2y)^2-4x^3y}}{2x},
$$
which generates a shift of the sequence \href{http://oeis.org/A089732}{\underline{A089732}} in~\cite{oeis}.
\end{thm}
\begin{proof} If a Dyck path with air pockets equals  $\alpha UD$ with $\alpha\in\aa\cup\{\varepsilon\}$, then its contribution to $P(x,y)$ is $(1+P(x,y))x^2y$; if it has the form $\alpha\beta^\sharp$ with $\beta\in\aa$, then its contribution is $(1+P(x,y))xP(x,y)$. Hence, the second-to-last return decomposition yields:
$$
P(x,y)=(1+P(x,y))(x^2y+xP(x,y)),
$$
which gives the result after solving for $P(x,y)$.
\end{proof}

\begin{cor}
The g.f.\ for the popularity of peaks in $\aa_n$  is
$$
\frac{x\left(1+x-x^2-\sqrt{x^4-2x^3-x^2-2x+1}\right)}{2\sqrt{x^4-2x^3-x^2-2x+1}},
$$
which generates a shift of the sequence \href{http://oeis.org/A203611}{\underline{A203611}} in~\cite{oeis}. An asymptotic approximation of the $n$-th term is $$\frac{\sqrt{5}-2 }{\sqrt{\pi n}\, \sqrt{14 \sqrt{5}-30}}\left(\frac{3+\sqrt{5}}{2}\right)^n,
$$
and an asymptotic for the expectation of the peak number is 
$$\frac{5-\sqrt{5}}{10}n\sim 0.276393191\cdot n.$$
\label{cor5}
\end{cor}

\begin{rem}
Another way of finding the total number of peaks in all Dyck paths with air pockets of length $n$ is the following: since the number of $n$-length Dyck paths with air pockets with $k$ peaks is $N(n-k,k)$, we have:
$$
\st{Peak}(\aa_n)=\sum_{k=1}^{\lfloor\frac{n}{2}\rfloor} kN(n-k,k)=\sum_{k=1}^{\lfloor\frac{n}{2}\rfloor}\frac{k}{n-k}\binom{n-k}{k}\binom{n-k}{k-1}.
$$

Using the formula for the sequence \href{http://oeis.org/A203611}{\underline{A203611}} in~\cite{oeis}, we get the following identity:
$$
\sum_{k=1}^{\lfloor\frac{n}{2}\rfloor}\binom{n-k-1}{k-1}\binom{n-k}{k-1}=\sum_{k=0}^{n-1}\binom{k-1}{2k-n}\binom{k}{2k-n+1}.
$$
\end{rem}

\subsection{The number of returns to the \texorpdfstring{$x$}{x}-axis}
A return to the $x$-axis is a  step $D_m$, $m\geq 1$, ending on the $x$-axis.
\begin{thm}
Let $R(x,y)=\sum_{n,k\geq 0} r_{n,k}x^ny^k$ be the g.f.\ where $r_{n,k}$ is the number of $n$-length Dyck paths with air pockets with $k$ returns, then:
$$
R(x,y)=\frac {2}{2-y\left(1-x+{x}^{2}-\sqrt {{x}^{4}-2\,{x}^{3}-{x}^{2}-2\,x+1}\right)}-1,
 $$ which generates the triangle \href{https://oeis.org/A098086}{\underline{A098086}} in~\cite{oeis} where the row $n$ and column $k$ gives  the number of peakless Motzkin paths having its leftmost $F$-step on the $k$-th step (see also Proposition~1).
\end{thm}

\begin{proof}  If a Dyck path with air pockets equals  $\alpha UD$ with $\alpha\in\aa\cup\{\varepsilon\}$, then its contribution to $R(x,y)$ is $(1+R(x,y))x^2y$; if it has the form $\alpha\beta^\sharp$ with $\beta\in\aa$, then its contribution is $(1+R(x,y))xyA(x)$. So we deduce, 
$$
R(x,y)=(1+R(x,y))(x^2y+xyA(x)),
$$
 which gives the result using relation~(\ref{rel1}).
\end{proof}

\begin{cor}
The g.f.\ for the popularity of returns to the $x$-axis in $\aa_n$  is $$2\frac {1-x+{x}^{2}-\sqrt {{x}^{4}-2\,{x}^{3}-{x}^{2}-2\,x+1}}{
 \left(1+x-{x}^{2}+\sqrt {{x}^{4}-2\,{x}^{3}-{x}^{2}-2\,x+1}\right) ^{2}},$$ which corresponds to the sequence \href{http://oeis.org/A093128}{\underline{A093128}} in~\cite{oeis}, where the $n$-th term counts all possible dissections of a regular $(n+2)$-gon using zero or more strictly disjoint diagonals. An asymptotic approximation of the $n$-th term is $$\frac{\sqrt{14 \sqrt{5}-30}\, \sqrt{5}}{4 n \sqrt{\pi n}}\left(\frac{3+\sqrt{5}}{2}\right)^{n+1},$$
 and an asymptotic for the expectation of the return number is $\sqrt{5}.$
 \label{co6}
\end{cor}

\subsection{The number of catastrophes}

A catastrophe is a step $D_m$, $m\geq 2$, ending on the $x$-axis.

\begin{thm}
Let $C(x,y)=\sum_{n,k\geq 0} c_{n,k}x^ny^k$ be the g.f.\ where $c_{n,k}$ is the number of $n$-length Dyck paths with air pockets with $k$ catastrophes. Then we have:
$$C(x,y)=\frac{2}{2-2x^2-y\left(1-x-x^2-\sqrt{x^4-2x^3-x^2-2x+1}\right)}-1.
$$
\end{thm}
\begin{proof}  If a Dyck path with air pockets equals  $\alpha UD$ with $\alpha\in\aa\cup\{\varepsilon\}$, then its contribution to $C(x,y)$ is $(1+C(x,y))x^2$; if it has the form $\alpha\beta^\sharp$ with $\beta\in\aa$, then its contribution is $(1+C(x,y))xyA(x)$. So, we deduce
$
C(x,y)=(1+C(x,y))(x^2+xyA(x)).
$
\end{proof}

\begin{cor}
The g.f.\ for the popularity of catastrophes in $\aa_n$  equals 
$$
2\frac{1-x-x^2-\sqrt{x^4-2x^3-x^2-2x+1}}{\left(1+x-x^2+\sqrt{x^4-2x^3-x^2-2x+1}\right)^2}.
$$
 An asymptotic approximation of the $n$-th term is $$\frac{\sqrt{14 \sqrt{5}-30}\,  \left(4-\sqrt{5}\right) }{4 n \sqrt{\pi n}}\left(\frac{3+\sqrt{5}}{2}\right)^{n+1},$$
and  an asymptotic for the expectation of the catastrophe number is $4-\sqrt{5}.$
\label{co7}
\end{cor}

\begin{rem} As a byproduct of Corollaries~\ref{co6} and~\ref{co7}, the ratio of the popularity of catastrophes in $\aa_n$ to the popularity of returns in $\aa_n$ tends to $\frac{4-\sqrt{5}}{\sqrt{5}}=0.788854\ldots$ when $n$ tends toward $\infty$.
\end{rem}

\subsection{The number of pyramids
\texorpdfstring{$U^kD_k$}{UkDk}}
A $k$-pyramid $\Delta_k$ in a path is an occurrence of the pattern $U^kD_k$, $k\geq 1$.

\begin{thm}
For all $k\geq 1$, the  g.f.\  $P_k(x,y)=\sum_{n,m\geq 0} p^k_{n,m}x^ny^m$ where $p^k_{n,m}$ is the number of $n$-length Dyck paths with air pockets having $m$ $k$-pyramids is given by:
$$
P_k(x,y)={\frac {{x}^{k+1}(y-1)-2\,{x}^{k+2}(y-1)+{x}^{2}+x-1+\sqrt {Q}}{2({x}^{k+2}(y-1)-x)}},
$$ where $$Q={x}^{k+1}\left(y-1\right)(x^{k+1}(y-1)+4x+2(x^2-x-1))+x^4-2x^3-x^2-2x+1.$$
\end{thm}

\begin{proof}
We refine the  first return decomposition so that any Dyck path with air pockets falls into one of the following cases: $(i)$~$\Delta_m\gamma$ with $1\leq m\leq k-1$, $(ii)$~$\Delta_k\gamma$, $(iii)$~$\Delta_{k+1}\gamma$, $(iv)$~$(\alpha \Delta_k)^\sharp\gamma$ with $\alpha\in\aa$, $(v)$~$\beta^\sharp\gamma$ with $\beta\in\aa$ not being $\Delta_m$ with $1\leq m\leq k$, nor  $\alpha\Delta_k$ with $\alpha\in\aa$, where $\gamma\in\aa\cup\{\varepsilon\}$. These five cases are disjoint and cover all Dyck paths with air pockets. So, we deduce:
$$\scriptsize
P_k=\left(\underbrace{\sum_{i=2}^{k}x^i}_{(i)}+\underbrace{\vphantom{\sum_{i=2}^{k}x^i}x^{k+1}y}_{(ii)}+\underbrace{\vphantom{\sum_{i=2}^{k}x^i}x^{k+2}}_{(iii)}+\underbrace{\vphantom{\sum_{i=2}^{k}x^i}x^{k+2}P_k}_{(iv)}+\underbrace{x\cdot
\left(P_k-\sum_{i=2}^{k}x^i -x^{k+1}y(1+P_k)\right)}_{(v)}\right)(1+P_k),
$$
where $P_k$ stands for $P_k(x,y)$. Solving for $P_k$, we get the result.
\end{proof}

\begin{cor} For $k\geq 1$, the g.f.\ for the popularity  $\st{\Delta_k}(\aa_n)$ of $k$-pyramids in $\aa_n$  equals:
$$
Y_k(x)=\frac{x^{k+1} \left(1 + 2 x^{2} - x^{3} +(1-x) \sqrt{x^4-2x^3-x^2-2x+1} \right)}{2 \sqrt{x^4-2x^3-x^2-2x+1}},
$$ which generates the $(n-k-2)$-th term of the sequence \href{https://oeis.org/A051291}{\underline{A051291}} in~\cite{oeis}.
In particular, we have $\st{\Delta_1}(\aa_n)=\st{\Delta_k}(\aa_{n+k-1})$ for all $k\geq 1$ and $n\geq 2$, which means that there are as many $1$-pyramids in $\aa_n$ as there are $k$-pyramids in $\aa_{n+k-1}$. An asymptotic approximation of the $n$-th term of this sequence is
$$
\frac{\sqrt{5}-1}{2\sqrt{\pi n}\sqrt{14\sqrt{5}-30}}\left(\frac{3+\sqrt{5}}{2}\right)^{n-k-1},
$$
and for the expected number of $k$-pyramids we have 
$$\frac{5-\sqrt{5}}{10}\left(\frac{3-\sqrt{5}}{2}\right)^{k}\cdot n.$$
\label{popupyramides}
\end{cor}

Notice that $Y_1(x)$ corresponds to the generating function for the popularity of down-steps $D$ (see Corollary~4), since each $D$ is necessarily preceded by an up-step. Moreover, we have $Y_k(x)=x^{k-1} Y_1(x)$ since each pyramid $\Delta_k$ in a path of length $n$ comes from a pyramid $\Delta_1$ in a path of length $n-k+1$ by adding $k-1$ up-steps and by increasing the length of the down-step.
An immediate consequence of Corollary~\ref{popupyramides} is the following.

\begin{cor} For $k\geq 1$, the g.f.\ for the popularities  $\st{\Delta_{\geq k}}(\aa_n)$ and $\st{\Delta_{\leq k}}(\aa_n)$ are respectively  given by $$Y_{\geq k}(x)=\frac{x^{k-1}}{1-x}Y_1(x) \mbox{ and } Y_{\leq k}(x)=\frac{1-x^k}{1-x}Y_1(x),$$
which means that $\st{\Delta_{\leq k}}(\aa_{n-k+1})=\st{\Delta_{\geq k}}(\aa_n)-\st{\Delta_{\geq k}}(\aa_{n-k}).$
\label{corpyrleast}
\end{cor}

\bigskip

For any $k \ge 1$, the popularity of pyramids of size at
least $k$ in $\aa_n$ (see Corollary~\ref{corpyrleast}) seems to correspond to a shift of the sequence
\oeis{A201631} in~\cite{oeis}, which enumerates {\em Fibonacci
  meanders with central angle $180$ degrees} (see
  Luschny's~\cite{Luschny} and Wienand's~\cite{Wienand} posts in OEIS Wiki
about meanders). In order to prove this fact,
we give the formal
definition of such meanders, and provide 
their g.f.\ that does
not exist in the literature (to our knowledge).

\newcommand\arcR{\tikz{\draw[thick,->] (0,0) arc  (180:0:2mm)} }
\newcommand\arcL{\tikz{\draw[thick,->] (0,0) arc  (0:180:2mm)} }
A Fibonacci meander with  central angle 180 degrees is a closed smooth
self-overlapping curve in the plane, 
consisting of an even length sequence of two types of arcs of
angle $180$ degrees, namely $L = \, \arcL$ and
$R = \arcR$, starting at the origin with an $L$-arc toward the north, having no consecutive $L$-arcs except at the beginning where a run (of any length) of
consecutive $L$-arcs is authorized. Each arc starts at the end of 
the previous arc and it preserves the direction of its arrow, i.e.\  $LLR$ corresponds to
\tikz{
\draw[thick,->,shorten >=0] (0,0) arc (0:180:2mm);
\draw[thick,->,shorten >=0] (-4mm,0) arc (180:360:2mm);
\draw[thick,->] (0,0) arc (180:0:2mm);
     minimum size=2] at (0,0) {};
}.
Let $\mathcal{F}_{2n}$ be the
set of such meanders of length $2n$. For instance, the left part of
Figure~\ref{benotafraid} illustrates a meander in $\mathcal{F}_{20}$.

\begin{figure}[H]
  \centering
  \includegraphics[width=32em]{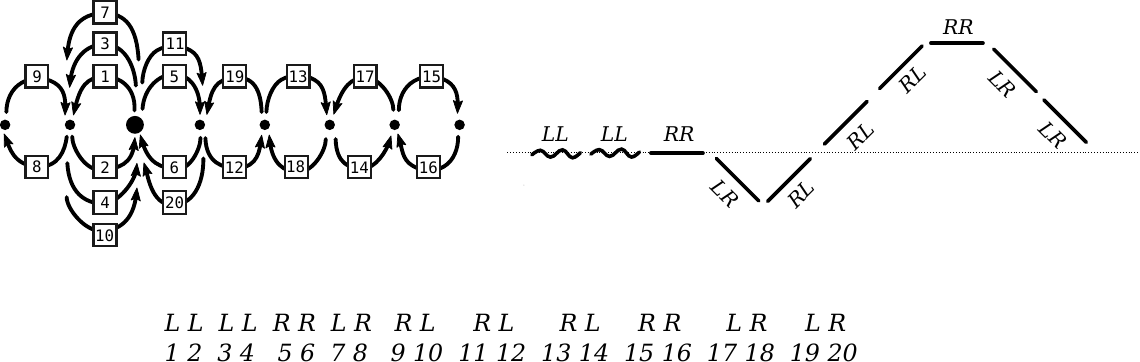}
  \caption{A Fibonacci meander in $\mathcal{F}_{20}$ where the sequence of arcs is given by $LLLLRRLRRLRLRLRRLRLR$, and its associated lattice path.
  }
  \label{benotafraid}
\end{figure}

Now, we define a function $\tau$, mapping a two-letter word over the
alphabet $\{L,R\}$ into the set $\{U,D,F,\widetilde{F}\}$:
$$
\tau(a) =
\begin{cases}
  U,&\text{if } a = RL,\\
  D,&\text{if } a = LR,\\
  F,&\text{if } a = RR,\\
  \widetilde{F}, & \text{if } a = LL,
\end{cases}
$$
and a function $\mu$, mapping a meander $w = w_1 w_2 \ldots
w_{2n}\in \mathcal{F}_{2n}$ into an $n$-length word over the alphabet
$\{U,D,F,\widetilde{F}\}$: $$\mu(w) = \tau(w_1 w_2) \; \tau(w_3 w_4)
\; \ldots \tau(w_{2n-1} w_{2n}).$$

 {\em Grand Motzkin paths} are lattice paths of length $n$ in
$\mathbb{N}\times \mathbb{Z}$, consisting
of steps $U = (1,1), D = (1,-1)$, and $F = (1,0)$, starting at
$(0,0)$, ending at $(n,0)$ (contrarily to classical Motzkin paths,
they can go below the $x$-axis). For instance $DUUDDFDUUUD$ is a grand
Motzkin path of length $11$. Let $\mathcal{G}_n$ be the set of peakless (i.e.\ with no occurrence of $UD$) grand Motzkin paths of length $n$ and $\mathcal{N}_n$ be the subset of paths starting with a $D$-step in $\mathcal{G}_n$.
Denote by $\mathcal{S}_n$ the set of peakless grand Motzkin paths of length $n$   starting with a down step or prefixed by a nonempty
sequence of special flats, called {\em wavy flats}, $\widetilde{F} = (1,0)$. For instance, $\widetilde{F}\widetilde{F}DUUDDFDUUUD\in \mathcal{S}_{13}$.

\begin{prop}
  The function $\mu$ induces a bijection from $\mathcal{F}_{2n}$ to $\mathcal{S}_n$.
  \label{mea-bij}
\end{prop}

\begin{proof} Every meander $a \in \mathcal{F}_{2n}$ avoids
  the pattern $LL$ except if the occurrence of $LL$ is only
  preceeded by letters $L$, which means that $\mu(a)$ avoids the
  pattern $UD$. If the meander $a$ starts with a maximal prefix of the form
  $(LL)^{k}$, $k\geq 1$, then $\mu(a)$ starts, with a maximal sequence of $k$ wavy flats, $\widetilde{F}^k$.  If the meander $a$ starts with $LR$, then  $\mu(a)$ starts with $D$. Moreover, the image by $\mu$ of a factor $RR$ is $F=\tau(RR)$ in $\mu(a)$, the image of $LR$ is $D$ and the image of $RL$ is $U$. Thus,  the fact that $a$ is a closed curve, implies that $\mu(a)$ starts and ends on $x$-axis.  Due to all these observations, $\mu$ is necessarily a bijection from $\mathcal{F}_{2n}$ to $\mathcal{S}_n$.
\end{proof}

 Notice that the bijection $\mu$ is very close to the bijection of Roitner (see \cite{Roit}) between $2$-watermelons with arbitrary deviation and weighted Motzkin paths, which suggests that there are tight links between $2$-watermelons and Wienand-Luschny meanders. It would be interesting to explore this correspondence in future works.
\begin{thm}
  The g.f.\ $S(x) = \sum_{n\geq 0} s_n x^n$,
  where the coefficient $s_n$ is the number of $2n$-length Fibonacci meanders
  with a central angle 180 degrees, is
  $$
  S(x)  =
  \frac{x^2 - x +1- \sqrt{R} }{(x-1)\left(R+(x^2-x-1)\sqrt{R}\right)},
  $$
  with $ R  = x^4-2x^3-x^2-2x+1$. Using Corollary~\ref{corpyrleast},  we have 
  $$S(x)=\frac{Y_1(x)}{x^2(1-x)}-1,$$ which establishes the expected  link between  Fibonacci meanders and the popularity $\st{\Delta_{\geq k}}(\aa_n)$.\label{mea-gf}
\end{thm}

\begin{proof}
  Considering Proposition~\ref{mea-bij}, it suffices to enumerate $\mathcal{S}_n$. We set $\mathcal{G}=\bigcup_{n\geq 0} \mathcal{G}_n$, $\mathcal{N}=\bigcup_{n\geq 0} \mathcal{N}_n$, and $\mathcal{S}=\bigcup_{n\geq 0} \mathcal{S}_n$.
  Recall that $\mathcal{M}$ is the set of nonempty peakless Motzkin paths.
  Denote by $\mathcal{V}$ the set of Motzkin paths without valleys $DU$ and  by $\widebar{\mathcal{V}}$ the set of paths obtained by 
  symmetry about the $x$-axis ($U\leftrightarrow D$) of valleyless Motzkin paths, e.g.\ $DUFDU\in \widebar{\mathcal{V}}$ 
  since it is symmetric to $UDFUD \in \mathcal{V}$.  Let $\mathcal{W}$
  be the set of  nonempty sequences of wavy flat steps. We use
  $M(x)$, $V(x) = \widebar{V}(x)$, $G(x)$, $N(x)$, $W(x)$ to
  denote the corresponding generating functions with respect to the length.

   Obviously, we have $W(x)=\frac{x}{1-x}$.
    From relation~(\ref{rel1}) and Theorem~\ref{bij} we obtain $M(x)=A(x)/x$. 
Also, there is a one-to-one correspondence $\nu$ between $\mathcal{M}_n$ and $\mathcal{V}_{n-1}$ that can be defined recursively by $\nu(F)=\varepsilon$, $\nu(FQ)=F\nu(Q)$, $\nu(UQD)=U\nu(Q)D$, and   $\nu(UQDR)=U\nu(Q)DF\nu(R)$ if $R$ is non-empty. So, we have $V(x)=M(x)/x=A(x)/x^2$. Finally, we decompose $\mathcal{G}$, $\mathcal{N}$ and $\mathcal{S}$ as illustrated below:

  \newcommand\N{\mathcal{N}}
  \newcommand\G{\mathcal{G}}
  \newcommand\W{\mathcal{W}}
  \newcommand\Sc{\mathcal{S}}
  \newcommand\M{\mathcal{M}}
  \newcommand\V{\mathcal{V}}
  \newcommand\du{\;\biguplus\;}
  \tikzset{STEP/.style={very thick, scale=0.3}}
  \newcommand\wavy{\tikz{\draw[STEP,decorate, decoration = {snake, segment length=1.4mm, amplitude=0.3mm}] (0,0) -- (1,0);}}
  \newcommand\stepflat{\tikz{\draw[STEP] (0,0) -- (1,0);}}
  \newcommand\stepup{\tikz{\draw[STEP] (0,0) -- (1,1);}}
  \newcommand\stepdown{\tikz{\draw[STEP] (0,0) -- (1,-1);}}
  \newcommand{\rise}[1]{\raisebox{2.05ex}{#1}}
  \newcommand{\swoop}[1]{\raisebox{-6.3ex}{#1}}
  \newcommand{\upcircle}[1]{
    \tikz{
      \draw[orange, very thick] (0,0) .. controls (0.2,0.7) and (1,1) .. (1.4,0) -- (0,0);
      \node at (0.7,0.3) {$#1$};
    }
  }
  \newcommand{\downcircle}[1]{
    \tikz{
      \draw[orange, very thick] (0,0) .. controls (0.2,-0.7) and (1,-1) .. (1.4,0) -- (0,0);
      \node at (0.7,-0.3) {$#1$};
    }
  }
  \newcommand{\aroundcircle}[1]{
    \raisebox{-5ex}{
    \tikz{
      \draw[orange, very thick] (0,0) .. controls (0.2,0.7) and (1,1) .. (1.4,0);
      \draw[orange, very thick] (0,0) .. controls (0.2,-0.5) and (1,-0.8) .. (1.4,0);
      \node at (0.7,0) {$#1$};
    }}
  }

  \begin{center}
    \scalebox{0.80}{
    $
  \begin{aligned}
    & \G = \varepsilon
           \du \stepflat \aroundcircle{\G}
           \du \stepup \rise{\upcircle{\M}}\stepdown \aroundcircle{\G}
           \du \aroundcircle{\N},\\
    & \N = \stepdown \swoop{\downcircle{\widebar{\V}}} \stepup
           \du \stepdown \swoop{\downcircle{\widebar{\V}}} \stepup
               \rise{\stepup \rise{\upcircle{\M }} \stepdown \aroundcircle{\G}}
           \du \stepdown \swoop{\downcircle{\widebar{\V}}} \stepup \rise{\stepflat \aroundcircle{\G}}
           ,\\
    & \Sc = \upcircle{\W} \aroundcircle{\G} \du \aroundcircle{\N},\\
  \end{aligned}
  $
  }
    \end{center}

  \noindent which induces the following system of functional equations
  $$
  \begin{cases}
    G(x) & = 1 + x G(x) + x^2 M(x) G(x) + N(x),\\
    N(x) & = x^2 V(x) + x^4 V(x) M(x) G(x) + x^3V(x) G(x),\\
    S(x) & = W(x) G(x) + N(x).
  \end{cases}
  $$
  Solving this system, we obtain  $S(x)$.
\end{proof}

\section{Non-decreasing Dyck paths with air pockets}

A Dyck path with air pockets is \emph{non-decreasing} if the sequence of heights of its valleys is non-decreasing, i.e.\ the sequence of the minimal ordinates of the occurrences $D_kU$, $k\geq 1$, is non-decreasing from left to right.  See~\cite{Barc} for a reference about non-decreasing Dyck paths. For example, the Dyck path with air pockets $UUDUDUD_2$ is non-decreasing, since its two valleys both lie at height $1$, while the path  $UUDUD_2UD$ is not, since its two valleys lie at heights $1$ and $0$ from left to right. Let $\aa_n^\inc$, $n\geq 2$,
be the set of $n$-length non-decreasing Dyck paths with air pockets and $\aa^\inc=\bigcup_{n\geq 2}\aa_n^\inc$. The subset of $n$-length prime non-decreasing Dyck paths with air pockets is defined as the intersection $\pp_n^\inc:=\aa_n^\inc\cap\pp$, and we set $\pp^\inc:=\bigcup_{n\geq 2}\pp_n^\inc$. Analogous to generic Dyck paths with air pockets, the map $\alpha\mapsto\alpha^\flat$ induces a bijection between $\pp_n^\inc$ and $\aa_{n-1}^\inc$, whose inverse is the map $\alpha\mapsto\alpha^\sharp$.

\begin{thm}
For $n\geq 2$, if $a^\inc_n$ is the number of $n$-length non-decreasing Dyck paths with air pockets, then $a_2^\inc =1$ and  $a_n^\inc =2^{n-3}$ for $n\geq 3$.
\end{thm}

\begin{proof}
Any non-decreasing Dyck path with air pockets $\alpha$ has one
of the following two forms:  ($i$)
 $\alpha\in \pp^\inc \cup \{UD\}$,
or ($ii$)  
$\alpha=\Delta_k\beta$ where $k\geq 1$ and 
$\beta\in\aa^\inc$. So, if 
$A^\inc(x)=\sum_{n\geq 2}a_n^\inc x^n$ where 
$a_n^\inc$ is the cardinality of $\aa_n^\inc$,
then the previous decomposition implies the
functional equation
$$
A^\inc(x)=x(A^\inc(x)+x)+\frac{x^2}{1-x}A^\inc(x).
$$
Thus we have
$
A^\inc(x)=\frac{x^2(1-x)}{1-2x}
$ which completes the proof.
\end{proof}

\subsection{The numbers of
\texorpdfstring{$U$}{U} and
\texorpdfstring{$D$}{D}}

\begin{thm} Let $A^\inc(x,y,z)=\sum_{n,k,\ell\geq 0} a_{n,k,\ell}^\inc x^ny^kz^\ell$  be the trivariate g.f.\ where $a_{n,k,\ell}^\inc$ is the number of $n$-length non-decreasing Dyck paths with air pockets having $k$ up-steps $U$ and $\ell$ down-steps $D$. Then, $A^\inc(x,y,z)$ equals
$$
\frac{x^2y(1-xy)(xyz-xy-z)(x^2yz+xy-1)}{(x^3y^2(z-1)+x^2y(y-z)-2xy+1)(x^3y^2(z-1)-x^2yz-xy+1)}.
$$
\end{thm}

\begin{proof}
Let $Z(x,y,z)=\sum_{n,k,\ell\geq 0} z_{n,k,\ell} x^ny^kz^\ell$, where $z_{n,k,\ell}$ is the number of $n$-length non-decreasing Dyck paths with air pockets having only valleys at height $0$, $k$ up-steps $U$ and $\ell$ down-steps $D$. Such a path has the form $UD\alpha$ or $\Delta_k\alpha$ with $k\geq 2$, where $\alpha$  has all its valleys at height 0. Then, we have 
$$
Z(x,y,z)=(1+Z(x,y,z))\left(x^2yz+\frac{x^3y^2}{1-xy}\right).
$$
Solving for $Z(x,y,z)$, we get:
$$
Z(x,y,z)=\frac{x^2 y (x y (1-z) + z)}{x^3 y^2 (z - 1) - x^2 y z - x y + 1}.
$$
Now, any non-decreasing Dyck path with air pockets belongs to one of the following cases: ($i$) $\beta UD$, ($ii$) $\beta(UD)^\sharp$, ($iii$) $\beta(\alpha UD)^\sharp$ ($\alpha$ having all its valleys at height~$0$), ($iv$) $\beta\alpha^\sharp$ ($\alpha$ having all its valleys at height 0, and not ending with $UD$), where
$\beta$ is either empty or has all its valleys at height~$0$. Thus, we have (for short, we use $A^\inc$ and $Z$ instead of $A^\inc(x,y,z)$ and $Z(x,y,z)$):
$$
A^\inc=(1+Z)(x^2yz+x^3y^2+x^3y^2Z+xy(A^\inc-x^2yz(1+Z))).
$$
Solving for $A^\inc$, we get the result.
\end{proof}

\begin{por}
For all $n\geq 1$, the number of $n$-length non-decreasing Dyck paths with air pockets
which have 
all valleys at height $0$ 
is equal to $F_{n-1}$, where $F_k$ is the $k$-th Fibonacci number.
\label{fib}
\end{por}

\begin{proof}
Plugging $y=z=1$ into the trivariate g.f.\ $Z$ in the proof of the previous theorem we obtain $Z(x,1,1)=\frac{x^2}{1-x-x^2}$, which is the g.f.\   for the right shift of the sequence of Fibonacci numbers.
\end{proof}

As we have made in Section~3.1,  we deduce the following.
\begin{cor} For all $k\geq 1$, the number of non-decreasing Dyck paths with air pockets:

$\bullet$ having $n$ up-steps $U$ is the $k$-th term of the sequence \href{http://oeis.org/A001519}{\underline{A001519}};

$\bullet$  having $k$ up-steps $U$ and no down-steps $D$ is the $(k-1)$-th term of the sequence \href{http://oeis.org/A099036}{\underline{A099036}}.
\end{cor}
\begin{proof}
We calculate $Z(1,y,1)$ and $Z(1,y,0)$, respectively.
\end{proof}

\begin{cor}
The popularity of up-steps $U$ in $\aa_n^\inc$ is equal to the $(n-2)$-th term of the sequence \href{http://oeis.org/A098156}{\underline{A098156}} in~\cite{oeis}. An asymptotic for the expectation of the up-step number is $(3n-2)/4$. 
\end{cor}
\begin{proof}
We calculate $\partial_y(Z(x,y,1))\rvert_{y=1}$.
\end{proof}

\begin{cor}\label{popuDnondecr}
The g.f.\ for the popularity 
of down-steps $D$ in $\aa_n^\inc$ equals:
$$\frac{x^2(1-x)(1-4x+5x^2-2x^3+x^5)}{(1-2x)^2(1-x-x^2)}.
$$ An asymptotic approximation of the $n$-th term is $n\cdot 2^{n-6}$. An asymptotic for the expectation of the down-step number is $n/8$.
\end{cor}
\begin{proof}
We calculate $\partial_z(Z(x,1,z))\rvert_{z=1}$.
\end{proof}

\subsection{The number of peaks}

\begin{thm}
For all $n\geq 2$ and $k\geq 1$, the number of $n$-length non-decreasing Dyck paths with air pockets having $k$ peaks is equal to $\binom{n-2}{2(k-1)}.$
\end{thm}
\begin{proof}
Let $B(x,y)$ be the g.f.\ where the coefficient of $x^ny^k $ is the number of $n$-length paths in $\aa^\inc$ having $k$ peaks. Any non-decreasing  Dyck path with air pockets is  either of the form $\Delta_1=UD$, or $\alpha^\sharp$, or $\Delta_k\beta$ with $k\geq 1$, with $\alpha,\beta\in\aa^\inc$. This yields the following functional equation:
$$
B(x,y)=x^2y+xB(x,y)+\frac{x^2}{1-x}yB(x,y)
$$
with the solution $B(x,y)=\frac{(1-x)x^2y}{(1-x)^2-x^2y}$, which generates the sequence \href{https://oeis.org/A034839}{\underline{A034839}} in~\cite{oeis}.
\end{proof}

\begin{cor} The popularity of peaks in $\aa_n^\inc$ is the $(n-2)$-th term of the sequence \href{https://oeis.org/A045891}{\underline{A045891}} in~\cite{oeis}, which is equal to $(n+2)\cdot2^{n-5}$ for $n\geq 4$. Then, the expectation of the peak number is $(n+2)/4.$
\end{cor}

\subsection{The number of returns to the \texorpdfstring{$x$}{x}-axis}

\begin{thm}  The bivariate g.f.\  where the coefficient of $x^ny^k$ is  the number of $n$-length non-decreasing Dyck paths with air pockets having $k$ returns  is $$
R(x,y)=\frac{x^2y(1-x)(1-x-x^2)}{(1-2x)(1-x-x^2y)}.
$$
\label{thm12}
\end{thm}
\begin{proof} Using the second-to-last return decomposition of $\aa^\inc$, we easily get the following functional equation:
$$
R(x,y)=x^2y+x^2yR(x,y)+xyA^\inc(x)+\frac{x^3}{1-x}yR(x,y),$$
which gives the result.
\end{proof}

\begin{cor} The g.f.\ for the popularity of returns in $\aa_n^\inc$ is $$
\frac{x^2(1-x)^2}{(1-2x)(1-x-x^2)},
$$ and for $n\geq 2$ the coefficient of $x^n$ is $2^{n-2}-F_{n-2}$, where $F_n$ is the $n$-th Fibonacci number (see \href{https://oeis.org/A099036}{\underline{A099036}} in~\cite{oeis}). Then, the expectation of the return number is $2-F_{n-2}/2^{n-3}$ that tends to $2$.
\label{co14}
\end{cor}

\subsection{The number of catastrophes}
\begin{thm}  The bivariate g.f.\ where the coefficient of $x^ny^k$ is  the number of $n$-length non-decreasing Dyck paths with air pockets having $k$ catastrophes is $$
C(x,y)=\frac{x^2(1-x)(1+x(y-2)-x^2y)}{(1-2x)(1-x-x^2-x^3(y-1))}.
$$
\end{thm}
\begin{proof}
First, let us determine the bivariate g.f.\ $U(x,y)$ with respect to the length and  number of catastrophes for non-decreasing Dyck paths with air pockets having all their valleys at height $0$. It is easy to see that
$$
U(x,y)=(1+U(x,y))\left(x^2+\frac{x^3y}{1-x}\right),
$$
which yields:
$$
U(x,y)=\frac{x^2(1-x+xy)}{1-x-x^2-x^3(y-1)}.
$$
Then, any non-decreasing Dyck path with air pockets has one of the following forms: 
($i$) $\beta UD$, or ($ii$) $\beta U^2D_2$, or ($iii$) $\beta \alpha^\sharp$ with $\alpha$ not belonging to forms ($i$) or ($ii$), and where $\beta$ is either empty or a non-decreasing Dyck path with air pockets which only has valleys that lie at height $0$. Hence, the bivariate generating function $C(x,y)$ satisfies the following equation:

$$
C(x,y)=(1+U(x,y))\left(x^2+x^3y(1+U(x,1))+xy\left(C(x,1)-x^2(1+U(x,1))\right)\right),
$$

which gives the result.
\end{proof}

\begin{cor} The g.f.\ for the popularity of catastrophes in $\aa_n^\inc$ equals 
$$\frac{x^3(1-x)(1-x+x^2)}{(1-2x)(1-x-x^2)},$$ and for $n\geq 4$ the coefficient of $x^n$ is $3\cdot 2^{n-4}+2F_{n-3}$, where $F_n$ is the $n$-th Fibonacci number (see the sequence \href{https://oeis.org/A175657}{\underline{A175657}} in~\cite{oeis}). Then, the expectation of the catastrophe number is $3/2-F_{n-3}/2^{n-4}$ that tends to $3/2$.
\label{co15}
\end{cor}

\begin{rem} As a byproduct of Corollary~\ref{co14}  and Corollary~\ref{co15}, the ratio of the popularity of catastrophes in $\aa_n^\inc$ to the popularity of returns in $\aa_n^\inc$ tends to $\frac{3}{4}$ when $n$ tends toward $\infty$.
\end{rem}

\subsection{The number of pyramids}

\begin{thm}
For $k\geq 1$, let $P_k(x,y)=\sum_{n,m\geq 0} p_{n,m}^kx^ny^m$ be the g.f.\ where $p_{n,m}^k$ is the number of $n$-length non-decreasing Dyck paths with air pockets with $m$ occurrences of the pattern $\Delta_k=U^kD_k$. Then the following holds:
$$
P_k(x,y)=\frac{x^2\left(1-\frac{x^2}{1-x}+x^{k-1}\left(1-x-\frac{x^2(2-x)}{1-x}\right)(y-1)-x^{2k}(y-1)^2\right)}{\left(1-x-\frac{x^2}{1-x}-x^{k+1}(y-1)\right)\left(1-\frac{x^2}{1-x}-x^{k+1}(y-1)\right)}.
$$
\end{thm}

\begin{proof} 
First, let us determine the expression of the bivariate g.f.\ $Z_k(x,y)$ with respect to the length and the number of patterns $\Delta_k$ for  non-decreasing Dyck paths with air pockets having all their valleys at height $0$. The second-to-last return decomposition of $\aa^\inc$ yields:
$$
Z_k(x,y)=(1+Z_k(x,y))\left(x^{k+1}y+\left(\frac{x^2}{1-x} - x^{k+1}\right)\right).
$$
Hence, we get
$
Z_k(x,y)=\frac{1}{1-\frac{x^2}{1-x}-x^{k+1}(y-1)}-1.
$

Now, assuming $k\neq 1$, any Dyck path has one of the following forms: ($i$) $\beta UD$, ($ii$) $\beta\Delta_{k-1}^\sharp$, ($iii$) $\beta\Delta_k^\sharp$, ($iv$) $\beta(\alpha\Delta_k)^\sharp$ with $\alpha$ having all of its valleys at height $0$, ($v$) $\beta\alpha^\sharp$ with $\alpha\in\aa^\inc$, $\alpha\neq \Delta_{k-1}, \Delta_k, \gamma\Delta_k$ ($\gamma$ having all of its valleys at height $0$), and where $\beta$ is either empty or has all its valleys at height $0$. This yields (for short we use $P_k$ and $Z_k$ instead of $P_k(x,y)$ and $Z_k(x,y)$):
$$
P_k=(1+Z_k)\left(x^2+x^{k+1}y+x^{k+2}(1+Z_k)+x\left(P_k-x^k-x^{k+1}y(1+Z_k)\right)\right).
$$
Solving for $P_k$, we get the result for $k\geq 2$.

If $k=1$, the expression of $P_1(x,y)$ is the same as that of the bivariate g.f.\ associated to the pattern $D$ in $\aa^\inc$ (given in Theorem~\ref{thm12}), because $D$ occurs exactly as often as $UD=\Delta_1$.
\end{proof}

\begin{cor}
For $k\geq 1$, the g.f.\ for the popularity $\st{\Delta_k}(\aa_n^\inc)$ of $k$-pyramids  in $\aa^\inc_n$ is
$$
W_k(x)=\frac{x^{k+1}(1-x)(1-4x+5x^2-2x^3+x^5)}{(1-2x)^2(1-x-x^2)}.
$$
In particular, we can see that $\st{\Delta_1}(\aa_n^\inc)=\st{\Delta_k}(\aa_{n+k-1}^\inc)$, which means that there are as many $1$-pyramids in $\aa_n^\inc$ as there are $k$-pyramids in $\aa_{n+k-1}^\inc$. An asymptotic approximation of the $n$-th term is $n\cdot 2^{n-5-k},$ and an asymptotic for the expectation of the $k$-pyramid number is $n/2^{k+2}.$
\end{cor}

An immediate consequence of the previous 
corollary is the following one, which is the $\aa_n^\inc$-counterpart of Corollary~\ref{corpyrleast}.

\begin{cor} 
For $k\geq 1$, the g.f.\ for the popularities 
$\st{\Delta_{\geq k}}(\aa_n^\inc)$ and $\st{\Delta_{\leq k}}(\aa_n^\inc)$ are respectively given by
$$W_{\geq k}(x)=\frac{x^{k-1}}{1-x} W_1(x) \mbox{ and } W_{\leq k}(x)=\frac{1-x^k}{1-x} W_1(x),$$
which means that $\st{\Delta_{\leq k}}(\aa_{n-k+1}^\inc)=\st{\Delta_{\geq k}}(\aa_n^\inc)-\st{\Delta_{\geq k}}(\aa_{n-k}^\inc).$
\label{cor17}
\end{cor}

\begin{table}[H]
    \centering
    \begin{tabular}{lll}
    \hline
    Pattern &Pattern popularity in $\aa_n^\inc$ &OEIS\\
    \hline
    $U$ & $1, 2, 5, 13, 32, 76, 176, 400, 896, 1984$ & \href{https://oeis.org/A098156}{\underline{A098156}}\\
    $D$ & $1, 0, 2, 3, 7, 15, 33, 72, 157, 341$ & \\ 
    Peak&$1, 1, 3, 7, 16, 36, 80, 176, 384, 832$&\href{http://oeis.org/A045891}{\underline{A045891}}\\
    Ret&$1, 1, 3, 6, 13, 27, 56, 115, 235, 478$&\href{http://oeis.org/A099036}{\underline{A099036}}\\
    Cat&$0, 1, 1, 4, 8, 18, 38, 80, 166, 342$&\href{https://oeis.org/A175657}{\underline{A175657}}\\ 
    $\Delta_k$ & $\underbrace{0, \hdots, 0}_{k-1\text{ zeroes}}, 1, 0, 2, 3, 7, 15, 33, 72, 157, 341$ &\\
    $\Delta_{\geq k}$ & $\underbrace{0, \hdots, 0}_{k-1\text{ zeroes}}, 1, 1, 3, 6, 13, 28, 61, 133, 290, 631$ & New ($=v_n$)\\
        $\Delta_{\leq k}$ & $\Delta_{\leq 1}\quad1, 0, 2, 3, 7, 15, 33, 72, 157, 341$& $v_n-v_{n-k}$\\
    &$\Delta_{\leq 2}\quad1, 1, 2, 5, 10, 22, 48, 105, 229, 498$&\\
    &$\Delta_{\leq 3}\quad1, 1, 3, 5, 12, 25, 55, 120, 262, 570$&\\
    &$\vdots$&\\
    \hline
    \end{tabular}
    \caption{Pattern popularity in $\aa^\inc_n$ for $2\leq n\leq 11$.}
\end{table}

 \paragraph{Going further.}

It should be interesting to give natural bijections whenever our enumerative results suggest such bijections.
Also, asymptotic investigations of expectations 
could be extended to a study of the limit 
distributions. It will also be of interest to 
investigate the `Grand' counterpart of Dyck 
paths with air pockets, that are paths where negative ordinates are allowed,
or `Motzkin' counterpart where 
flat steps $(1,0)$ are allowed.

\section{Acknowledgments}

The authors would like to thank the anonymous referees for useful
remarks and comments.  This work was supported in part by the project
ANER ARTICO funded by Bourgogne-Franche-Comté region of France and ANR-22-CE48-0002.

\end{document}